\documentclass[3p,times,a4paper]{my-elsarticle}

\usepackage{a4wide}

\usepackage[colorlinks=true,linkcolor=black,citecolor=blue,urlcolor=blue]{hyperref}

\journal{Applied Mathematics Letters}

\usepackage{amsmath,amssymb,amsthm}
\usepackage{subfig}
\usepackage{graphicx}
\usepackage{epstopdf}

\newtheorem{theorem}{Theorem}[section]

\theoremstyle{remark}

\theoremstyle{definition}
\newtheorem{definition}{Definition}[section]

\def\R{\mathbb{R}}

\begin{document}

\begin{frontmatter}
	
\title{A stochastic SICA epidemic model for HIV transmission}

\author[Serbia]{Jasmina Djordjevic}
\ead{djordjevichristina@gmail.com}

\author[CIDMA]{Cristiana J. Silva}
\ead{cjoaosilva@ua.pt}

\author[CIDMA]{Delfim F. M. Torres\corref{correspondingauthor}}
\cortext[correspondingauthor]{Corresponding author.}
\ead{delfim@ua.pt}

\address[Serbia]{Faculty of Science and Mathematics, University of Ni\v{s},
Vi\v{s}egradska 33, 18000 Ni\v{s}, Serbia}

\address[CIDMA]{Center for Research and Development in Mathematics and Applications (CIDMA),\\
Department of Mathematics, University of Aveiro, 3810-193 Aveiro, Portugal}

% ---------------------------------------------
	
\begin{abstract}
We propose a stochastic SICA epidemic model for HIV transmission,  
described by stochastic ordinary differential equations, and discuss 
its perturbation by environmental white noise. 
Existence and uniqueness of the global positive solution 
to the stochastic HIV system is proven, and conditions under 
which extinction and persistence in mean hold, are given. 
The theoretical results are illustrated via numerical simulations.
\end{abstract}


\end{frontmatter}

% ---------------------------------------------
	
\section{Introduction}

Epidemics are, inevitably, affected by environmental white noise, which is 
an important component to be taken into account by mathematical models,
providing an additional degree of realism in comparison to their 
deterministic counterparts \cite{ZhaoJiang:ApMathLet:2014}.
Here, our aim is to improve the deterministic 
SICA epidemic model for HIV transmission recently proposed 
in \cite{SilvaTorres:EcoComplexity,SilvaTorres:DCDS:2017},
by considering environmental interactions.
For that, we follow \cite{ZhaoJiang:ApMathLet:2014,Grafton:BMB:2005,%
Gray:SIAM:JAM:2011,Greenhalgh:PhysA:2016,Lu:PhysA:2005,Tornatore}
and introduce stochastic noise in the form of a 
Brownian motion with positive intensity.
The advantage of our model with respect to previous ones in 
\cite{SilvaTorres:EcoComplexity,SilvaTorres:DCDS:2017}
is that we assume fluctuations in the environment, manifesting 
in the transmission coefficient rate, thus making it more biologically 
realistic for the transmission dynamics of HIV/AIDS in a homogeneously 
mixing population of variable size.

The model subdivides human population into four mutually-exclusive
compartments: susceptible individuals ($S$);
HIV-infected individuals with no clinical symptoms of AIDS
(the virus is living or developing in the individuals
but without producing symptoms or only mild ones)
but able to transmit HIV to other individuals ($I$);
HIV-infected individuals under ART treatment (the so called
chronic stage) with a viral load remaining low ($C$);
and HIV-infected individuals with AIDS clinical symptoms ($A$).
The total population at time $t$, denoted by $N(t)$, is given by
$N(t) = S(t) + I(t) + C(t) + A(t)$.
For sake of simplicity, we assume that the associated 
AIDS-induced mortality is negligible. Using the same 
arguments as in \cite{SilvaTorres:DCDS:2017}, 
we consider a force of infection given by 
$\beta \left( I + \eta_C \, C  + \eta_A  A \right)$ 
with $\beta = \displaystyle \frac{\beta_0 \, \mu}{\Lambda}$, 
where $\beta_0$ is the effective contact rate for HIV transmission.
The modification parameter $\eta_A \geq 1$ accounts for the relative
infectiousness of individuals with AIDS symptoms, in comparison to those
infected with HIV with no AIDS symptoms. Individuals with AIDS symptoms
are more infectious than HIV-infected individuals (pre-AIDS) because
they have a higher viral load and there is a positive correlation
between viral load and infectiousness \cite{art:viral:load}.
On the other hand, $\eta_C \leq 1$ translates the partial restoration
of immune function of individuals with HIV infection
that use ART correctly \cite{AIDS:chronic:Lancet:2013}.
All individuals suffer from natural death, at a constant rate $\mu$.
We also assume that HIV-infected individuals, with and without AIDS 
symptoms, have access to ART treatment. HIV-infected individuals with 
no AIDS symptoms $I$ progress to the class of individuals with HIV 
infection under ART treatment $C$ at a rate $\phi$, and HIV-infected 
individuals with AIDS symptoms are treated for HIV at rate $\alpha$.
Individuals in the class $C$ leave to the class $I$ at a rate $\omega$.
Moreover, an HIV-infected individual with AIDS symptoms $A$
that starts treatment moves to the class of HIV-infected individuals $I$,
moving only to the chronic class $C$ if the treatment is maintained.
HIV-infected individuals $I$ with no AIDS symptoms, which do not take
ART treatment, progress to the AIDS class $A$ at rate $\rho$.
Precisely, we consider the model
\begin{equation}
\label{eq:det:model}
\begin{cases}
d S (t) = \left[ \Lambda - \beta \left( I(t) + \eta_C \, C(t)
+ \eta_A  A(t) \right) S(t) - \mu S(t) \right] dt,\\[0.2 cm]
d I (t) = \left[ \beta \left( I(t) + \eta_C \, C(t)
+ \eta_A  A(t) \right) S(t) - \xi_3 I(t) + \alpha A(t) + \omega C(t) \right] dt, \\[0.2 cm]
d C (t) = \left[ \phi I(t) - \xi_2 C(t) \right] dt,\\[0.2 cm]
d A (t) = \left[ \rho \, I(t) - \xi_1 A(t) \right] dt,
\end{cases}
\end{equation}
where $\xi_1 = \alpha + \mu + d$,
$\xi_2 = \omega + \mu$ and $\xi_3 = \rho + \phi + \mu$.
Existence and uniqueness of solution to the deterministic model 
\eqref{eq:det:model} is proved in 
\cite{SilvaTorres:EcoComplexity,SilvaTorres:DCDS:2017},
where it is shown that the system has one disease free equilibrium
when the basic reproduction number is less than one and one endemic
equilibrium when the basic reproduction number is greater than one.
Local and global stability of the equilibrium points of \eqref{eq:det:model}
is also proved in \cite{SilvaTorres:EcoComplexity,SilvaTorres:DCDS:2017}.
Motivated by \cite{Gray:SIAM:JAM:2011}, we consider here
fluctuations in the environment, which are assumed to
manifest themselves as fluctuations in the
parameter $\beta$, so that $\beta \rightarrow  \beta + \sigma \dot{B}(t)$, 
where $B(t)$ is a standard Brownian motion with intensity $\sigma^2 > 0$.
Our stochastic model takes then the following form:
\begin{equation}
\label{eq:model:stoch}
\begin{cases}
d S (t) = \left[ \Lambda - \beta \left( I(t) + \eta_C \, C(t)
+ \eta_A  A(t) \right) S(t) - \mu S(t) \right] dt 
- \sigma \left( I(t) + \eta_C \, C(t)
+ \eta_A  A(t) \right) S(t) dB(t),\\[0.2 cm]
d I (t) = \left[ \beta \left( I(t) + \eta_C \, C(t)
+ \eta_A  A(t) \right) S(t) - \xi_3 I(t) + \alpha A(t) + \omega C(t) \right] dt 
+ \sigma \left( I(t) + \eta_C \, C(t)  + \eta_A  A(t) \right) S(t) dB(t), \\[0.2 cm]
d C (t) = \left[ \phi I(t) - \xi_2 C(t) \right] dt,\\[0.2 cm]
d A (t) = \left[ \rho \, I(t) - \xi_1 A(t) \right] dt.
\end{cases}
\end{equation}

The paper is organized as follows: Section~\ref{sec:2} is devoted to
existence and uniqueness of a global positive solution
to the Stochastic Differential Equation (SDE) \eqref{eq:model:stoch}
(cf. Theorem~\ref{existence}); Section~\ref{sec:3} to conditions for the 
extinction of HIV within the population (cf. Theorem~\ref{extinction});
and Section~\ref{sec:4} to conditions for the persistence in mean
of the disease (cf. Theorem~\ref{persistence theor}). We end with
Section~\ref{sec:5}, illustrating both theoretical results of
extinction and persistence with numerical simulations.

% -----------------------------------

\section{Existence and uniqueness of a positive global solution}
\label{sec:2}

Throughout the paper, let $(\Omega, \mathcal{F}, \{\mathcal{F}\}_{t\geq0}, \mathcal{P})$ 
be a complete probability space with filtration $\{\mathcal{F}\}_{t\geq0}$,
which is right continuous and such that $\mathcal{F}$
contains all $\mathcal{P}$-null sets. The scalar Brownian motion $B(t)$ 
of \eqref{eq:model:stoch} is defined on the given probability space.
Also, we denote $\R_+^4=\left\{(x_1,x_2,x_3,x_4)|x_i>0, i=\overline{1,4}\right\}$.

\begin{theorem}
\label{existence}
For any $t \geq 0$ and any initial value $\left( S(0), I(0), C(0), A(0) \right) \in \R_+^4$, 
there is a unique solution $\left(S(t), I(t), C(t), A(t) \right)$ to 
the SDE \eqref{eq:model:stoch} and the solution remains in $\R_+^4$ 
with probability one. Moreover,
\begin{equation}
\label{eq:NgoesLam:div:mu}
N(t) \to \frac{\Lambda}{\mu} \mbox{ as } t \to \infty,
\end{equation}
where $N(t)= S(t)+I(t)+ C(t)+ A(t)$.
\end{theorem}

\begin{proof} 
Having in mind that  $N(t) = S(t) + I(t) + C(t) + A(t)$, we known that 
$A(t)=N(t)-S(t)-I(t)-C(t)\geq 0$, $t\geq 0$. 
It also follows that we can eliminate $A(t)$ from our SICA model 
\eqref{eq:model:stoch}, reducing it to a system of three equations: 
\begin{equation}
\label{eq:model:stoch:3}
\begin{cases}
d S (t) = \left[ \Lambda - \beta \left( I(t) + \eta_C \, C(t)
+ \eta_A  (N(t)-S(t)-I(t)-C(t)) \right) S(t) - \mu S(t) \right] dt \\
\phantom{d S (t) =} - \sigma \left( I(t) + \eta_C \, C(t)
+ \eta_A  (N(t)-S(t)-I(t)-C(t)) \right) S(t) dB(t),\\[0.1cm]
d I (t) = \left[ \beta \left( I(t) + \eta_C \, C(t)
+ \eta_A  (N(t)-S(t)-I(t)-C(t)) \right) S(t) - \xi_3 I(t) \right.\\
\left. \phantom{d S (t) =} + \alpha (N(t)-S(t)-I(t)-C(t)) + \omega C(t) \right] dt \\
\phantom{d S (t) =} + \sigma \left( I(t) + \eta_C \, C(t)  
+ \eta_A  (N(t)-S(t)-I(t)-C(t)) \right) S(t) dB(t), \\[0.1cm]
d C (t) = \left[ \phi I(t) - \xi_2 C(t) \right] dt.
\end{cases}
\end{equation}
If we prove that there exists a unique positive solution $(S(t), I(t), C(t))$ 
of system \eqref{eq:model:stoch:3} for $t\geq 0$, then we can replace process 
$I(t)$ in the last equation of system \eqref{eq:model:stoch} and solve it explicitly. 
From this fact, the existence of a unique positive solution for 
system \eqref{eq:model:stoch} is obtained. For a given 
$\left( S(0), I(0), C(0), A(0) \right) \in \mathbb{R}_+^4$, we prove 
that there exists a unique positive solution of system \eqref{eq:model:stoch:3} 
for every $t\geq 0$. Because the coefficients of system \eqref{eq:model:stoch:3} 
are locally Lipschitz continuous, there is a unique local solution 
on $[0,\tau_0)$ for any initial value $(S(0), I(0), C(0))$, where $\tau_0$ is known
in the literature as the \emph{explosion time}. It is necessary to prove that the solution 
is global, i.e., that $\tau_0=+\infty$ almost surely ($a.s.$, for brevity). Let us define
$$
\tau^+ = \inf\left\{t\in [0,\tau_0): S(t)\leq 0 \mbox{ or }  
I(t)\leq 0 \mbox{ or } R(t)\leq 0\right\}.
$$
Because the infimum of an empty set is $\infty$ and $\tau^+\leq\tau_0$, 
if we prove that $\tau^+=\infty\ a.s.$, then the proof of our theorem is complete. 
Indeed, if $\tau^+=\infty\ a.s.$, then $\tau_0=\infty$, which means that 
$(S(t), I(t), C(t))\in \mathbb{R}_+^3$ for $t\geq 0$ $a.s.$
Let us assume that  $\tau^+<\infty$. Then there exists $T>0$ such that 
$P(\tau^+<T)>0$. Define the function $V(t)=\ln\left(S(t)I(t)C(t)\right)$, 
which is twice differentiable and defined on positive values. 
By Ito's formula, 
\begin{equation*}
d(V(t))\geq K(S(t),I(t),C(t))+\sigma\left(\frac{S(t)}{I(t)}-1\right) 
\left( I(t) + \eta_C \, C(t) + \eta_A  (N(t)-S(t)-I(t)-C(t)) \right)  dB(t),
\end{equation*}
where 
\begin{equation*}
\begin{split}
&K(S(t),I(t),C(t)) = - \beta \left( I(t) + \eta_C \, C(t)
+ \eta_A  (N(t)-S(t)-I(t)-C(t)) \right)- \mu  \\
&\phantom{d(V(t))=} -\frac{1}{2}\sigma^2 \left( I(t) + \eta_C \, C(t)
+ \eta_A  (N(t)-S(t)-I(t)-C(t)) \right) S(t) - \xi_3\\
&\phantom{d(V(t))=}-\frac{ S^2(t)}{I^2(t)} \sigma ^2\left[ I(t) 
+ \eta_C \, C(t)  + \eta_A  (N(t)-S(t)-I(t)-C(t)) \right)^2].
\end{split}
\end{equation*}
By definition of our model, we are considering the case when infection is high.
Thus, $\displaystyle \frac{S(t)}{I(t)}-1\leq 0$ and 
\begin{equation}
\label{eq:rev:3}
\begin{split}
&V(t)\geq V(0)+\int_0^tK(S(s),I(s), C(s))ds\\ 
&\phantom{V(t)\leq}+\int_0^t\sigma\left(\frac{S(s)}{I(s)}-1\right) 
\left( I(s) + \eta_C \, C(s) + \eta_A  (N(s)-S(s)-I(s)-C(s)) \right)dB(s).
\end{split}
\end{equation}
It follows that $\lim_{t\rightarrow \tau^+} V(t)=-\infty$.
Letting $t\rightarrow \tau^+$ in \eqref{eq:rev:3}, we have
\begin{multline*}
-\infty \geq V(t) \geq V(0)+\int_0^tK(S(s),I(s), C(s))ds\\ 
+\int_0^t\sigma\left(\frac{S(s)}{I(s)}-1\right) 
\left( I(s) + \eta_C \, C(s) 
+ \eta_A  \left(N(s)-S(s)-I(s)-C(s)\right) \right)dB(s)
> -\infty,
\end{multline*}
which is in contradiction with the assumptions. We conclude that 
$\tau^+=\infty$ $a.s.$ 

It remains to prove \eqref{eq:NgoesLam:div:mu}.
If we sum all equations from system \eqref{eq:model:stoch}, then 
\begin{eqnarray*}
&&d(S(t)+I(t)+C(t)+A(t))=\left[\Lambda-\mu S(t) 
+\left( \phi- \xi_3+\rho\right) I(t) + \left(\alpha  - \xi_1 \right)A(t) 
+ \left(\omega- \xi_2\right) C(t)    \right]dt \\
&&\phantom{d(S(t)+I(t)+C(t)+A(t))}
=\left[\Lambda-\mu S(t) +( \phi-  \rho - \phi - \mu+\rho) I(t)\right.\\
&&\phantom{d(S(t)+I(t)+C(t)+A(t))=}\left.
+ \left(\alpha  - \alpha - \mu - d \right)A(t)
+ \left(\omega- \omega - \mu\right) C(t)\right]dt\\
&&\Leftrightarrow \frac{d(S(t)+I(t)+C(t)+A(t))}{dt}
=\Lambda-\mu(S(t)+I(t)+C(t)+A(t))-d\cdot A(t).
\end{eqnarray*}
Solving the last equation, we obtain that
$$
S(t)+I(t)+ C(t)+ A(t)=e^{-\mu t}\left[S(0)+I(0)+ C(0)+ A(0)
+\int_0^t(\Lambda -d\cdot A(t))e^{\mu s}ds\right].
$$
Applying L'Hospital's rule, it follows that
$\lim_{t\rightarrow +\infty}(S(t)+I(t)+ C(t)+ A(t))
=\displaystyle \frac{\Lambda}{\mu}$.
The proof is complete.
\end{proof} 

% ------------------------------------

\section{Extinction}
\label{sec:3}

In this section, we prove a condition 
for the extinction of the disease.

\begin{theorem}
\label{extinction}
Let $Y(t)=\left(S(t), I(t), C(t), A(t) \right)$ be the solution 
of system \eqref{eq:model:stoch} with positive initial value. 
Assume that $\sigma^2>\frac{\beta}{2 \xi_3}$. Then, 
$$
I(t),C(t),A(t)\rightarrow 0 \text{ a.s. and }
S(t)\to \displaystyle \frac{\Lambda}{\mu} \text{ a.s.},
$$
as $t \to +\infty$.
\end{theorem}

\begin{proof} 
By Theorem~\ref{existence}, the solution of system \eqref{eq:model:stoch} 
is positive for every $t\geq 0$. Applying It\^{o}'s formula on the second 
equation of system \eqref{eq:model:stoch}, we have
\begin{equation}
\label{e1}
\begin{split}
&d(\log I(t))
=\left\{-\frac{\sigma^2}{2}\left[\frac{S(t)}{I(t)}\left( I(t) 
+ \eta_C \, C(t)  + \eta_A  A(t) \right)-\frac{\beta}{\sigma^2} \right]^2
+\frac{\beta^2}{2\sigma^2}
+\frac{1}{I(t)}[- \xi_3 I(t) + \alpha A(t) + \omega C(t)] \right\}dt\\
&\phantom{d(\log I(t))=}+\frac{\sigma S(t)}{I(t)}
\left( I(t) + \eta_C \, C(t)  + \eta_A  A(t) \right)dB(s).
\end{split}
\end{equation}
Integrating both sides of \eqref{e1} from $0$ to $t$, 
and then dividing by $t$, we obtain that
\begin{equation}
\label{e2}
\frac{\log I(t)}{t}\leq  \frac{\log I(0)}{t}+\frac{\beta^2}{2\sigma^2}
- \xi_3+\frac{J(t)}{t}+\frac{M(t)}{t},
\end{equation}
where we define
$$
J(t)= \int_0^t\frac{\alpha A(s) + \omega C(s)}{I(s)}ds,
\quad M(t)=\int_0^t\frac{\sigma S(s)}{I(t)}\left( 
I(s) + \eta_C \, C(s)  + \eta_A  A(s) \right)dB(s).
$$
We need to estimate functions $J(t)$ and $M(t)$. 
As $N(t) = S(t) + I(t) + C(t) + A(t)$, then each coordinate 
of the population  ($S$, $I$, $C$, $A$) is less or equal than 
the number of the whole population $N(t)$. We have
$$
J(t) = \int_0^t\frac{\alpha A(s) 
+ \omega C(s)}{I(s)}ds\leq (\alpha+\omega)\frac{\Lambda}{\mu}.
$$
Because $M(t)$ is an integral with respect to the Brownian motion, 
it is local continuous martingale. Also, if we replace the upper bound 
with $t=0$ in $M(t)$, then we have $M(0)=0$. Further, we can find 
the quadratic variation and obtain the following limits:
$$
\limsup_{t\rightarrow +\infty}\frac{\langle M,M\rangle_t}{t}
\leq \frac{\sigma^2\Lambda^4(1+ \eta_C+\eta_A)^2 }{\mu^4}<+\infty.
$$
Applying the large number theorem for martingales \cite{Gray:SIAM:JAM:2011}, we have that
\begin{equation}
\label{eq:aux}
\lim_{t\rightarrow +\infty}\frac{M(t)}{t}=0 \ \ a.s.
\end{equation}
If we use \eqref{eq:aux} into estimates in \eqref{e2}, 
and the fact that $\sigma^2>\frac{\beta}{2 \xi_3}$, then 
$$
\limsup_{t\rightarrow +\infty}\frac{\log I(t)}{t}\leq \frac{\beta^2}{2\sigma^2}- \xi_3<0, \ a.s.
$$
This implies that $\lim_{t\rightarrow +\infty}I(t)=0$ a.s.
Solving explicitly the ordinary differential equation for process $C(t)$, 
from system \eqref{eq:model:stoch} we have
$$
C(t)=e^{-\xi_2t}\left[C(0)+\int_0^t\phi I(s)e^{\xi_2s} ds\right]
\leq  e^{-\xi_2t}C(0)+\phi\int_0^t I(s) ds.
$$
As $\lim_{t\rightarrow +\infty}I(t)=0$ a.s.,
we also have that $\lim_{t\rightarrow +\infty}C(t)=0$ a.s.
Similarly, we obtain that $\lim_{t\rightarrow +\infty}A(t)=0$ a.s.
Estimation of $S(t)$ is easy: since
$N(t)= S(t)+I(t)+ C(t)+ A(t)$ and 
from \eqref{eq:NgoesLam:div:mu} we know that
$N(t) \to \frac{\Lambda}{\mu}\ \mbox{ a.s. }$ as $t \to +\infty$, 
replacing $I(t),C(t),A(t)\rightarrow0$ a.s., $t\rightarrow+\infty$,
we obtain that $S(t)\to \frac{\Lambda}{\mu}$ a.s., $t \to +\infty$.
\end{proof}

% -----------------------

\section{Persistence in mean} 
\label{sec:4}

We begin by recalling the notion of persistence in mean.

\begin{definition}
\label{persistence}
System \eqref{eq:model:stoch} is said to be persistent in mean if 
$\displaystyle \lim_{t\rightarrow \infty}\frac{1}{t}\int_0^tI(s)ds>0$ a.s.
\end{definition}

Let us introduce the notation $[x(t)]=\displaystyle \frac{1}{t}\int_0^t x(s)ds>0$.

\begin{theorem}
\label{persistence theor} 
Let
\begin{equation}
\label{eq:K1}
K_1=\frac{\beta}{\mu}\left(\frac{\alpha\rho}{\xi_1}-\xi_3
+\frac{\omega\phi}{\xi_2}\right)+\frac{\mu(\xi_1\xi_2- \mu)}{\Lambda(1 
+ \eta_C + \eta_A )}.
\end{equation}
For any initial value $(S(0),I(0), C(0),A(0))\in \R_+^4$ such that 
$$
S(t)+I(t)+ C(t)+ A(t)=N(t) \to \frac{\Lambda}{\mu} \text{ as } t \to \infty, 
$$
if $K_1 \neq 0$, $\frac{1}{K_1}\left(\frac{\Lambda \beta}{\mu}-\xi_1\xi_2 
-\frac{\sigma^2\Lambda^2}{2\mu^2}\right) > 0$ and $\xi_1,\xi_2>1$, then 
the solution $(S(t),I(t), C(t),A(t))$ satisfies
$$
\liminf_{t\rightarrow \infty} [I(t)]
\geq\frac{1}{K_1}\left(\frac{\Lambda \beta}{\mu}
-\xi_1\xi_2 -\frac{\sigma^2\Lambda^2}{2\mu^2}\right).
$$
\end{theorem}

\begin{proof}
An integration of system \eqref{eq:model:stoch} yields
\begin{equation*}
\frac{S(t)-S(0)}{t}+\frac{I(t)-I(0)}{t}
+\frac{\alpha}{\xi_1}\frac{A(t)-A(0)}{t}
+\frac{\omega}{\xi_2}\frac{C(t)-C(0)}{t}
=\Lambda-\mu[S(t)]
+\left(\frac{\alpha\rho}{\xi_1}-\xi_3
+\frac{\omega\phi}{\xi_2}\right)[I(t)].
\end{equation*}
From here, one has
\begin{equation}
\label{p1}
[S(t)]=\frac{\Lambda}{\mu}+\frac{1}{\mu}\left(\frac{\alpha\rho}{\xi_1}-\xi_3
+\frac{\omega\phi}{\xi_2}\right)[I(t)]-\frac{K(t)}{\mu},
\end{equation}
where $K(t)=\displaystyle \frac{S(t)-S(0)}{t}+\frac{I(t)-I(0)}{t}+\frac{\alpha}{\xi_1}\frac{A(t)
-A(0)}{t}+\frac{\omega}{\xi_2}\frac{C(t)-C(0)}{t}$. By It\^{o}'s formula, we obtain
\begin{multline*}
d\log \left(I(t)+ \eta_C \, C(t)  + \eta_A  A(t)\right)\\
\geq \left\{\beta S(t) -\frac{ \mu I(t) + \xi_1\eta_A A(t)
+\xi_2\eta_C  C(t)}{I(t) + \eta_C \, C(t)
+ \eta_A  A(t)}-\frac{\sigma^2S^2(t)}{2(I(t) + \eta_C \, C(t)
+ \eta_A  A(t))^2}\right\}dt+ \sigma  S(t) dB(t).
\end{multline*}
We assume $\xi_1,\xi_2>1$. Then,	
\begin{eqnarray*}
&&\hskip-1cm d\log \left(I(t)+ \eta_C \, C(t)  + \eta_A  A(t)\right)\\
&&\hskip-1cm \phantom{d\log (I(t)+ \eta_C \, C(t)}\geq \left\{\beta S(t) 
-\frac{ \mu I(t) + \xi_1\xi_2\eta_A A(t)+\xi_1\xi_2\eta_C  C(t)}{I(t) + \eta_C \, C(t)
+ \eta_A  A(t)}-\frac{\sigma^2\Lambda^2}{2\mu^2}\right\}dt+ \sigma  S(t) dB(t)\\
&&\hskip-1cm \phantom{d\log (I(t)+ \eta_C \, C(t)}\geq \left\{\beta S(t)
-\xi_1\xi_2 +\frac{\mu(\xi_1\xi_2- \mu) I(t)}{\Lambda(1 + \eta_C \,
+ \eta_A )}-\frac{\sigma^2\Lambda^2}{2\mu^2}\right\}dt+ \sigma  S(t) dB(t).
\end{eqnarray*}
Integrating the last inequality from 0 to $t$, and dividing it with $t$, we have
\begin{eqnarray*}
&&\frac{\log (I(t)+ \eta_C \, C(t)  + \eta_A  A(t))
-\log (I(0)+ \eta_C \, C(0)  + \eta_A  A(0))}{t} \\
&&\phantom{aaaaa }\geq \beta [S(t)] +\frac{\mu(\xi_1\xi_2- \mu)}{\Lambda(1 + \eta_C \,
+ \eta_A )}[I(t)]-\xi_1\xi_2 -\frac{\sigma^2\Lambda^2}{2\mu^2}
+ \frac{\sigma}{t}\int_0^t  S(t) dB(t)\\
&&\phantom{aaaaa}\geq\frac{\Lambda \beta}{\mu}+\left[
\frac{\beta}{\mu}\left(\frac{\alpha\rho}{\xi_1}-\xi_3
+\frac{\omega\phi}{\xi_2}\right)+\frac{\mu(\xi_1\xi_2- \mu)}{\Lambda(1 + \eta_C \,
+ \eta_A )}\right][I(t)]-\xi_1\xi_2 -\frac{\sigma^2\Lambda^2}{2\mu^2}
-\frac{\beta K(t)}{\mu} +\frac{M(t)}{t}.
\end{eqnarray*}
It follows that
\begin{equation}
\label{p2}
\begin{split}
&[I(t)]\geq \frac{1}{K_1}\left[ \frac{\Lambda \beta}{\mu}
-\xi_1\xi_2-\frac{\sigma^2\Lambda^2}{2\mu^2}
-\frac{\beta K(t)}{\mu} +\frac{M(t)}{t} \right.\\
&\phantom{[I(t)]\geq \frac{1}{K_1}aa}\left. 
-\frac{\log (I(t)+ \eta_C \, C(t)  + \eta_A  A(t))
-\log (I(0)+ \eta_C \, C(0)  + \eta_A  A(0))}{t}\right],
\end{split}
\end{equation}
where $K_1$ is given by \eqref{eq:K1} and
$M(t)=\sigma \displaystyle \int_0^t  S(t) dB(t)$.
Process $M(t)$ is a local continuous martingale with value 0 for $t=0$, 
and it has the property that $\limsup_{t\rightarrow \infty}
\frac{\langle M,M\rangle_t}{t}\leq \frac{\sigma^2\Lambda^2}{\mu^2}<+\infty$,
a.s. Applying the large number theorem for martingales \cite{Gray:SIAM:JAM:2011}, 
it follows that $\lim_{t\rightarrow \infty}\frac{M(t)}{t}=0$, a.s.
As $S(t)+I(t)+ C(t)+ A(t)\leq \frac{\Lambda}{\mu}$,
$$
-\infty< \log \left(I(t)+ \eta_C \, C(t) + \eta_A  A(t)\right)
< \log\left(\frac{\Lambda}{\mu}(1+\eta_C+\eta_A)\right)
$$
and $\lim_{t\rightarrow \infty}K(t)=0$. By taking 
the limit inferior of both sides in Eq. \eqref{p2}, 
one has
$$
\liminf_{t\rightarrow \infty} [I(t)]
\geq\frac{1}{K_1}\left(\frac{\Lambda \beta}{\mu}
-\xi_1\xi_2 -\frac{\sigma^2\Lambda^2}{2\mu^2}\right),
$$
which completes the proof.
\end{proof}

% --------------------------------

\section{Numerical simulations}
\label{sec:5}

In this section we consider the following initial conditions: 
\begin{equation}
\label{eq:initi:cond}
S(0) = 5\, , \quad I(0) = 1\, , \quad C(0) = 0.1\, , \quad I(0) =0.1.
\end{equation}
We start by illustrating the extinction result proved in Theorem~\ref{extinction}.
The transmission coefficient is assumed to take the value $\beta_0 = 0.005$ and
the intensity of the Brownian motion to take the value $\sigma^2 = 0.01$. 
The inequality $\sigma^2>\frac{\beta}{2 \xi_3}$ from Theorem~\ref{extinction} 
is satisfied: $\sigma^2 - \frac{\beta}{2 \xi_3} = 0.009 > 0$. 
The rest of the parameters take the following values 
based on \cite{SilvaTorres:DCDS:2017} and references cited therein: 
$\mu = 1/69.54$, $\Lambda = 2.1 \mu$, $\eta_C = 0.015$, $\eta_A = 1.3$, 
$\phi = 1$, $\rho = 0.1$, $\alpha = 0.33$, and $\omega = 0.09$.  
The extinction of the stochastic model is observed numerically 
in Figure~\ref{fig:extinction}.
% ----------------------------------------
\begin{figure}[!htb]
\centering
\subfloat[\footnotesize{$S$: Susceptible}]{\label{S:ext}
\includegraphics[width=0.25\textwidth]{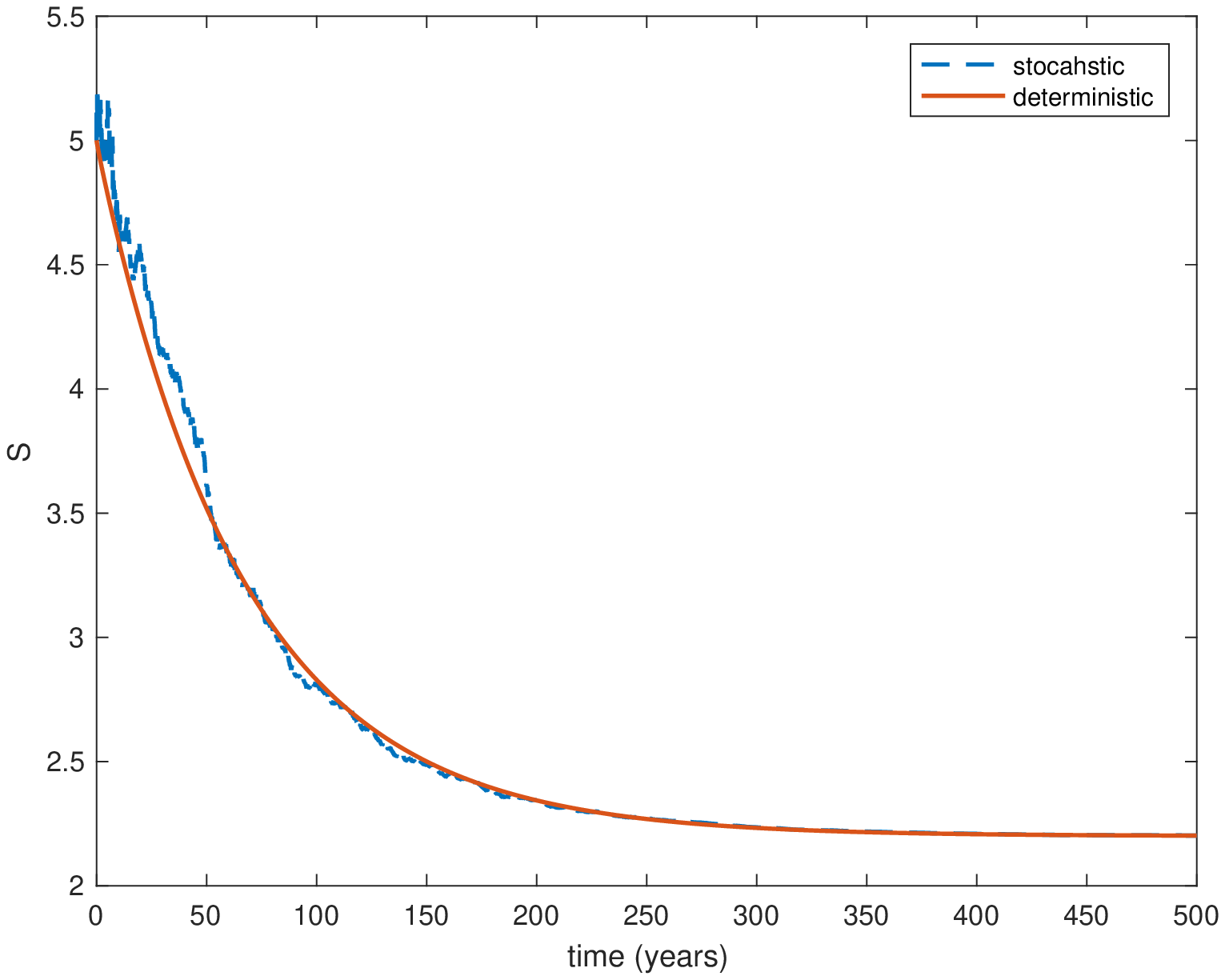}}
\subfloat[\footnotesize{$I$: Infected}]{\label{I:ext}
\includegraphics[width=0.25\textwidth]{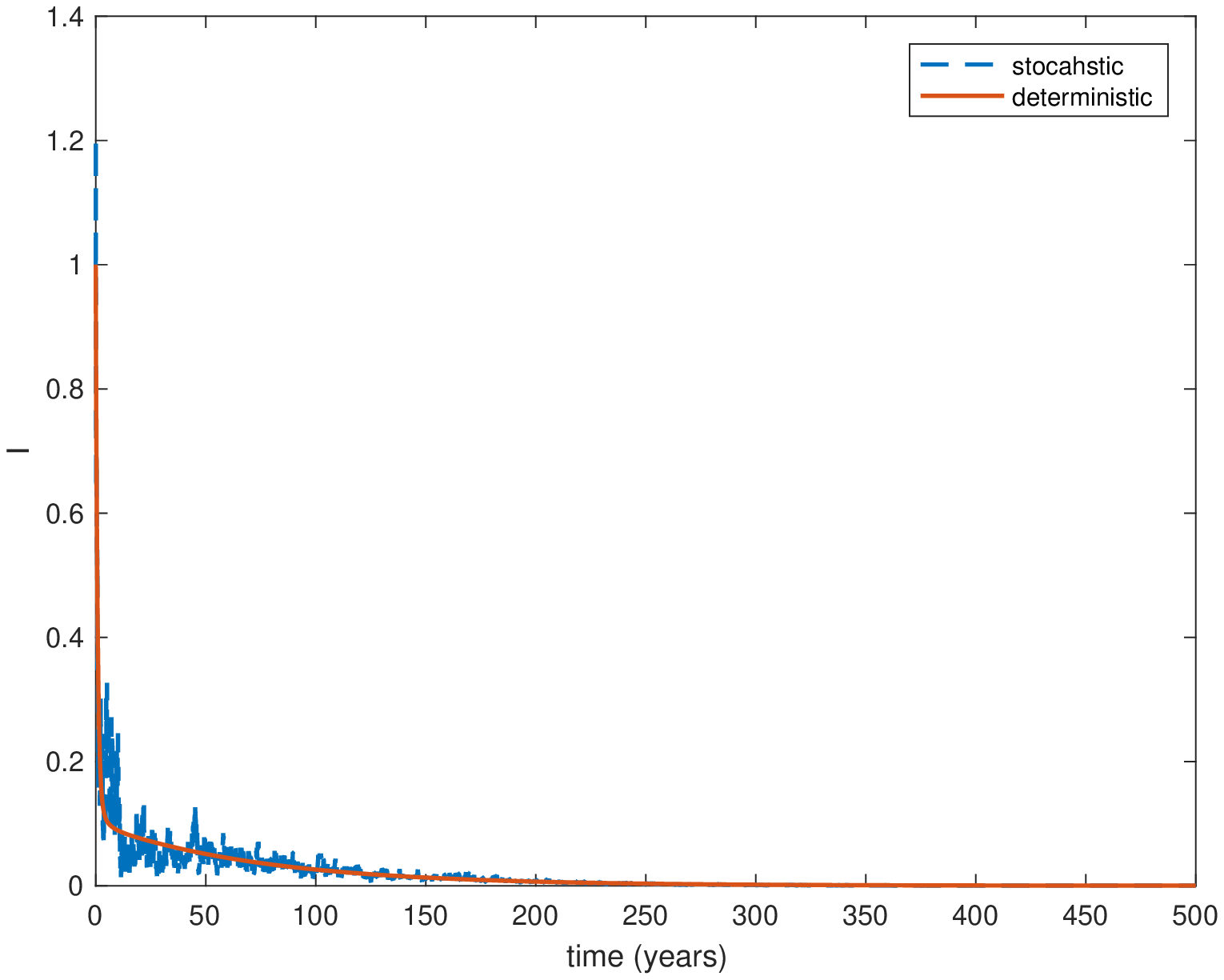}}
\subfloat[\footnotesize{$C$: Chronic}]{\label{C:ext}
\includegraphics[width=0.25\textwidth]{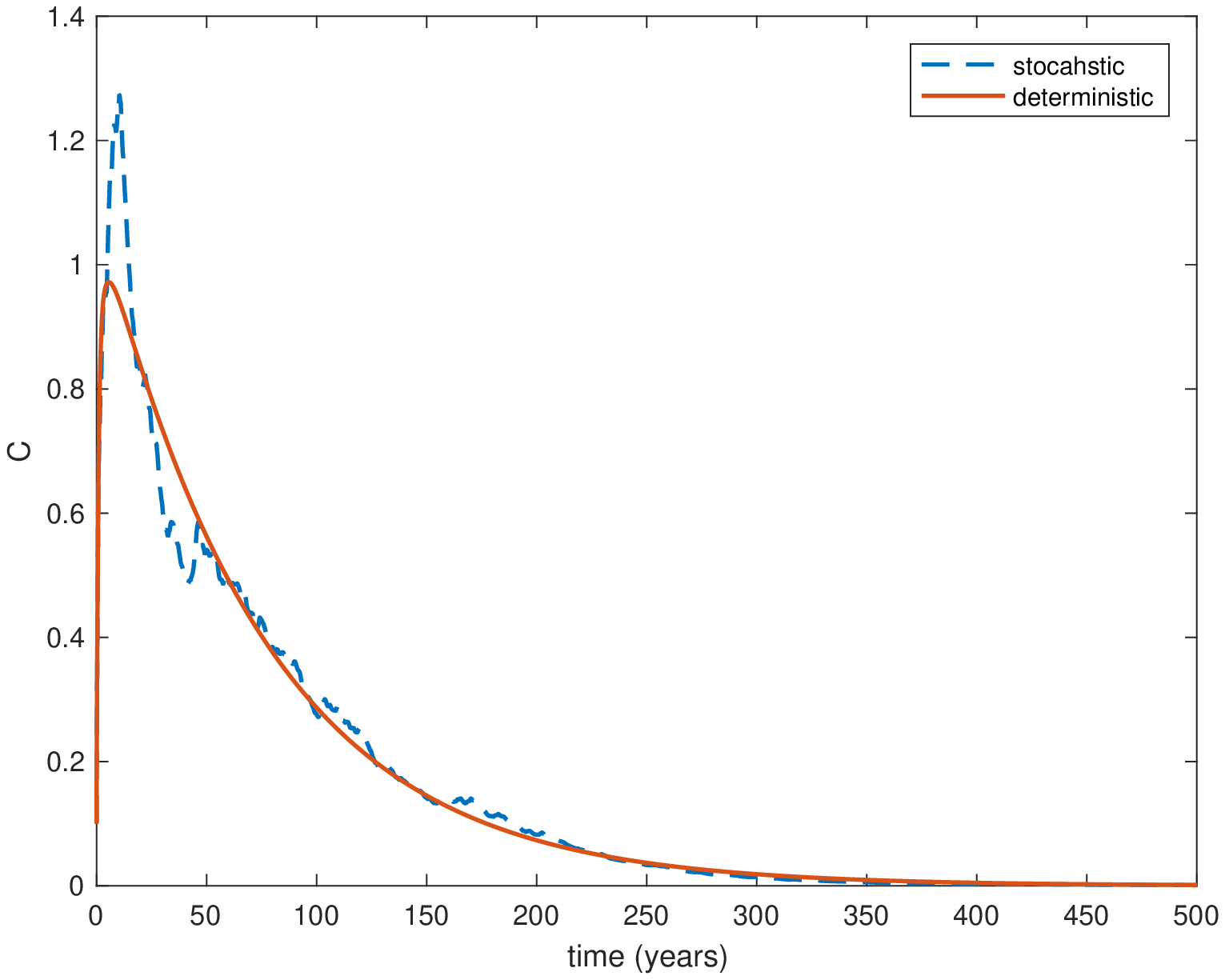}}
\subfloat[\footnotesize{$A$: AIDS}]{\label{A:ext}
\includegraphics[width=0.25\textwidth]{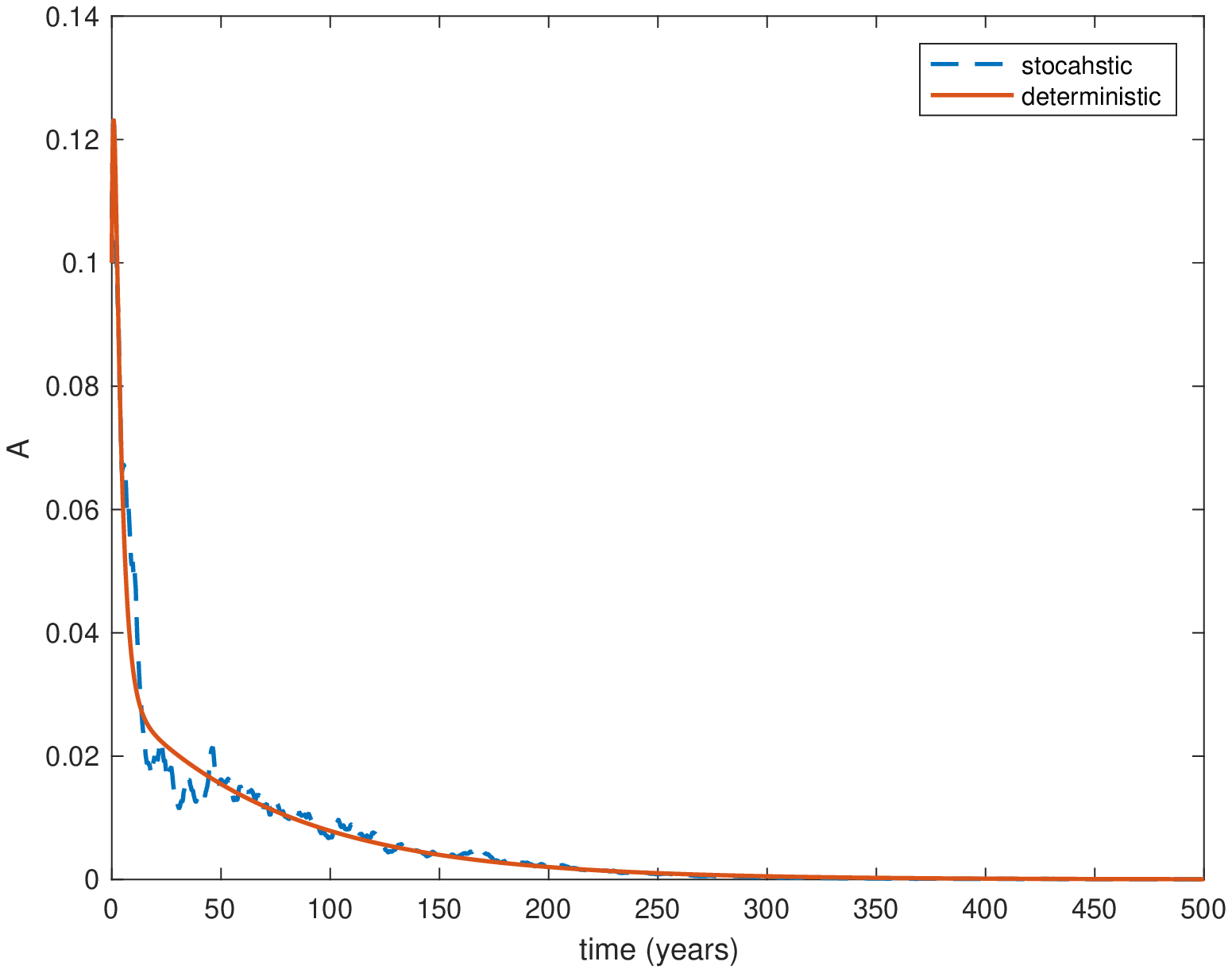}}	
\caption{Extinction: stochastic (dashed line)
and deterministic (continuous line) cases.}
\label{fig:extinction}
\end{figure}
% ----------------------------------------
To illustrate the persistence result proved in Theorem~\ref{persistence theor}, 
we consider the parameter values $\mu = 1/69.54$, $\Lambda = 2.1 \mu$, 
$\eta_C = 0.015$, $\eta_A = 1.3$, $\phi = 1$, $\rho = 0.1$, $\omega = 0.99$ and $\alpha = 0.99$. 
We take $\beta = 0.4$ and $\sigma^2 = 0.01$. For these parameter values, 
one has $\frac{1}{K_1}\left(\frac{\Lambda \beta}{\mu}-\xi_1\xi_2 
-\frac{\sigma^2\Lambda^2}{2\mu^2}\right) = 0.3012 > 0$. 
In Figure~\ref{fig:persistence}, we observe the persistence of the disease.
% --------------------------------------
\begin{figure}[!htb]
\centering
\subfloat[\footnotesize{$I(t)$, for $t \in [0, 1000]$.}]{\label{I:persistence}
\includegraphics[width=0.45\textwidth]{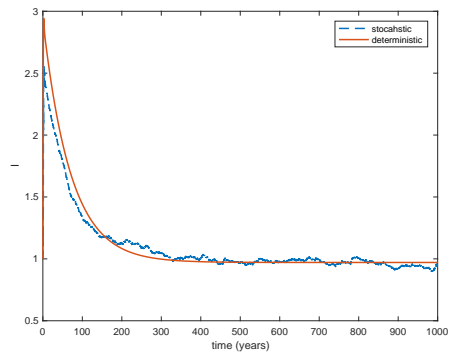}}
\subfloat[\footnotesize{$I(t)$, for $t \in [200, 2000]$.}]{\label{I:persistence:zoom}
\includegraphics[width=0.45\textwidth]{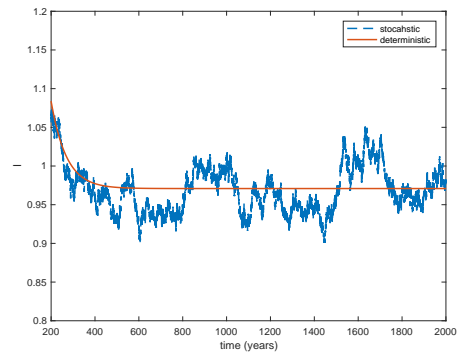}}
\caption{Infected individuals $I$: stochastic (dashed line) 
and deterministic (continuous line) cases.}
\label{fig:persistence}
\end{figure}

% -------------------------------------------------

\section*{Acknowledgments}

Jasmina Djordjevi\'c is supported 
by grant no.~174007 of MNTRS, 
while Silva and Torres are supported
by the Portuguese Foundation for Science and Technology (FCT)
within the R\&D unit CIDMA (UID/MAT/04106/2013) 
and TOCCATA research project PTDC/EEI-AUT/2933/2014.
Silva is also grateful to the FCT post-doc
fellowship SFRH/BPD/72061/2010.
The authors would like to thank a referee, 
for his/her valuable suggestions and comments.

% --------------------------

% -------------------

\end{document}